\documentclass [12pt,a4paper]{article}

\usepackage{amssymb} 

\usepackage{graphicx}

\usepackage{graphics}
\usepackage{amsmath}
\usepackage{amsthm}

\usepackage{enumerate}

\newtheorem{thm}{Theorem}

\newtheorem*{defi}{Definition}

\newtheorem*{cor}{Corollary}

\newtheorem*{prop}{Proposition}

\newtheorem{lemma}{Lemma}
\theoremstyle{definition}
\newtheorem*{ex}{Example}

\newtheorem*{remark}{Remark}

\def\qed{\nobreak\hfill $\square$}

\def\<{\langle}

\def\>{\rangle}

\def\eps{\varepsilon}

\def\iB{{\cal B}}

\def\bM{{\bf M}}

\def\D{{\mathbf{D}}}

\def\bM{{\mathbf{M}}}

\def\M3{M_3(\bbbc)}

\def\M3r{M_3(\bbbr)}

\def\bbbr{{\mathbb R}}

\def\bbbc{{\mathbb C}}

\newcommand{\R}{\mathbb{R}}

\newcommand{\N}{\mathbb{N}}

\newcommand{\C}{\mathbb{C}}

\newcommand*{\be}{\begin{equation}}

\newcommand*{\ee}{\end{equation}}

\newcommand*{\bit}{\begin{itemize}}
\newcommand*{\eit}{\end{itemize}}
\newcommand*{\ben}{\begin{enumerate}}
\newcommand*{\een}{\end{enumerate}}

\newcommand{\tr}{\mathrm{Tr}}

\newcommand{\mc}[1]{\mathcal{#1}}

\newcommand{\Hh}{\mathcal{H}}

\newcommand{\abs}[1]{\left|#1\right|}

\newcommand*{\inner}[2]{\left<#1,\,#2\right>}
\newcommand*{\braket}[2]{\left<#1\right.\left|\,#2\right>}

\newcommand*{\bra}[1]{\left<#1\right|}
\newcommand*{\ket}[1]{\left|#1\right>}
\newcommand*{\ler}[1]{\left(#1\right)}

\newcommand{\ba}{\begin{array}}
\newcommand{\ea}{\end{array}}

\newcommand*{\diff}[1]{\mathrm{d} #1}
\newcommand*{\dperd}[1]{\frac{\diff{}}{\diff{#1}}}
\newcommand*{\ddperd}[1]{\frac{\diff{}^2}{\diff{#1}}}
\newcommand*{\kiert}[2]{\left.#1 \right|_{#2}}

\begin{document}
\bigskip
\centerline{\LARGE {\bf On the joint convexity}}
\bigskip
\centerline{\LARGE {\bf of the Bregman divergence of matrices}}
\bigskip
\centerline{{\bf J\'ozsef Pitrik\footnote{E-mail: pitrik@math.bme.hu } and D\'aniel 
Virosztek\footnote{E-mail: virosz@math.bme.hu}}}

\begin{center}
Department of Mathematical Analysis, \\ Budapest University of Technology and Economics, \\ 
Egry J\'ozsef u.~1., Budapest, 1111 Hungary
\end{center}

\bigskip

\begin{abstract}
We characterize the functions for which the corresponding Bregman divergence is jointly convex on matrices. As an application of this characterization, we derive a sharp inequality for the quantum Tsallis entropy of a tripartite state, which can be considered as a generalization of the strong subadditivity of the von Neumann entropy. (In general, the strong subadditivity of the Tsallis entropy fails for quantum states, but it holds for classical states.) Furthermore, we show that the joint convexity of the Bregman divergence does not imply the monotonicity under stochastic maps, but every monotone Bregman divergence is jointly convex.\\
\medskip
{\bf Keywords:} joint convexity, Bregman divergence, Tsallis entropy, monotonicity. \\
{\it Mathematics subject classification (2010): 46N50, 46L30, 81Q10}
\end{abstract}

\section{Introduction}

In applications that involve measuring the dissimilarity between two objects (numbers, vectors, matrices, 
functions and so on) the definition of a divergence becomes essential. One such measure is a distance 
function, but there are many important measures which do not satisfy the properties of distance. For instance, 
the square loss function has been used widely for regression analysis, Kullback-Leibler divergence \cite{KL} has been applied to compare two probability density functions, the Itakura-Saito divergence \cite{IS} is used as a measure of the perceptual 
difference between spectra, or the Mahalonobis distance \cite{Mah} is to measure the dissimilarity between two 
random vectors of the same distribution. The Bregman divergence was introduced by 
Lev Bregman \cite{Breg} for convex functions $\phi :\R^d\to\R$ with gradient $\nabla\phi$, as the $\phi$-depending 
nonnegative measure of discrepancy
\be \label{Breg_origin}
D_\phi(p,q)=\phi(p)-\phi(q)-\langle\nabla\phi(q),p-q\rangle
\ee
of $d$-dimensional vectors $p,q\in\R^d$. Originally his motivation was the problem of convex programming, 
but it became widely researched both from theoretical and practical viewpoints. For example the remarkable fact that all the aforementioned divergences are special cases of the Bregman 
divergence shows its importance \cite{ba05}. In some literature it is applied under the name Bregman distance, in spite of that it is not in general the usual metric distance. Indeed, 
$D_\phi$ is reflexive but does not satisfy the triangle inequality nor symmetry. 
In addition to the wide range of applications in information theory, statistics and computer science, 
D\'enes Petz suggested the extension of the concept of Bregman divergence to operators \cite{pd07}. If $C$ 
denotes a convex set in a Banach space and $\iB(\Hh)$ denotes the bounded linear operators on the Hilbert space 
$\Hh$, for an operator valued smooth function $\Psi : C\to\iB(\Hh)$ the Bregman operator divergence is 
defined by 
\be \label{opdiv}
D_\Psi(x,y)=\Psi(x)-\Psi(y)-\lim_{t\to +0}\frac{\Psi(y+t(x-y))-\Psi(y)}{t}
\ee
for all $x,y\in C$. Since the Bregman operator divergence can be written as 
$$
D_\Psi(x,y)=\lim_{t\to +0}\frac{t\Psi(x)+(1-t)\Psi(y)-\Psi(tx+(1-t)y)}{t}
$$
for operator convex $\Psi$ functions $D_\Psi(x,y)\ge 0$ remains true for the standard partial ordering 
between self-adjoint operators. In this paper we investigate some interesting and important properties of 
the trace of Bregman operator divergence, and for our convenience we restrict ourself to matrices.
\par
Particularly we give a necessary and sufficient condition for the joint convexity of $\tr D_\Psi(x,y)$ and we investigate the relations between joint convexity and different notions of monotonicity. These properties are widely investigated and have several applications. For example, Tropp used the joint convexity of the quantum relative entropy - which is a special Bregman divergence - to give a succinct proof of a famous concavity theorem of Lieb \cite{tr12}. In \cite{ls14}, Lewin and Sabin characterized a certain monotonicity property of the Bregman divergence by the operator monotonicity of the derivative of the corresponding scalar function. In \cite{BaBo}, Bauschke and Borwein gave a necessary and sufficient condition for the joint convexity of the Bregman divergences on $\R^d.$ However, the question about the joint convexity of the trace of Bregman operator divergence has been left open.
\par
Throughout this paper the following notations will be used. $\R^+$ ($\R^{++}$) consists of all 
nonnegative (positive) numbers and $\bM_n$ ($\bM_n^{sa}, \bM_n^{+}, \bM_n^{++})$ denotes the set of $n \times n$ complex (self-adjoint, positive semidefinite, positive definite) matrices. Similarly, $\mc{B}(\mc{H})$ ($\mc{B}^{sa}(\mc{H}),$ $\mc{B}^{+}(\mc{H}),$ $\mc{B}^{++}(\mc{H})$) is the set of bounded (self-adjoint, positive semidefinite, positive definite) linear operators on the Hilbert space $\mc{H}.$ $\bM_n$ is endowed with the Hilbert-Schmidt inner product $\inner{X}{Y}=\tr X^{*} Y.$ If $f$ is an $\R \supset I \rightarrow \R$ function then the corresponding \emph{standard matrix function} is the following map:
$$
f: \{A \in \bM_n: \ \sigma(A) \subset I \} \rightarrow \bM_n, \, A=\sum_j \lambda_j P_j \mapsto f(A):=\sum_j f(\lambda_j) P_j,
$$
where $\sigma(A)$ is the spectrum and $\sum_j \lambda_j P_j$ is the spectral decomposition of $A.$

\subsection{Definition and basic properties}
Let $f: (0, \infty) \rightarrow \R$ be a convex function. Then the induced map
$$
\varphi_f: \ \bM_n^{++} \rightarrow \R, \ \ X \mapsto \varphi_f(X):= \tr f(X)
$$
is convex, as well \cite{Carlen}. A differentiable convex function is underestimated by its first-order Taylor polynomial, no matter what the base point is. Therefore, the expression
$$
\varphi_f(X)-\varphi_f(Y)- \D \varphi_f[Y](X-Y),
$$
where $\D \varphi_f[Y]$ denotes the Fr\'echet derivative of $\varphi_f$ at the point $Y,$
is nonnegative for any $X, Y \in \bM_n^{++}.$
By the linearity of the trace, for any $Y \in \bM_n^{++},$ $\D \varphi_f[Y]= \tr \circ \D f[Y],$ where $\D f[Y]$ denotes the Fr\'echet derivative of the standard matrix function $f: \bM_n^{++} \rightarrow \bM_n^{sa}$ at $Y.$ Let us define the central object of this paper precisely.

\begin{defi}
Let $f \in C^1 ((0, \infty))$ be a convex function and $X, Y \in \bM_n^{++}.$
The Bregman $f$-divergence of $X$ and $Y$ is defined by
\be \label{bregdef}
H_f (X,Y)=\tr \ler{f(X)-f(Y)-\D f[Y](X-Y)}.
\ee
\end{defi}
Note that this definition of the Bregman $f$-divergence coincides with the trace of the Bregman operator divergence (\ref{opdiv}), if $\Psi$ is the standard matrix function $f$ and $C=\bM_n^{++}, \, \mc{H}=\C^n.$ 

Consider the spectral decomposition $A=\sum_j \lambda_j \ket{\varphi_j}\bra{\varphi_j}$ 
of the positive definite matrix $A$ and denote the corresponding matrix units by $E_{ij}:=\ket{\varphi_i}\bra{\varphi_j}.$ The Fr\'echet derivative of the standard matrix function $f: \bM_n^{++} \rightarrow \bM_n^{sa}, \ X \mapsto f(X)$ at the point $A \in \bM_n^{++}$ is
\be \label{freform}
\D f [A]=\sum_{i,j} \int_{0}^{1} f'\ler{\lambda_j +t(\lambda_i-\lambda_j)} \mathrm{d}t \ket{E_{ij}}\bra{E_{ij}},
\ee
where Hermite's formula is used for the divided difference matrix \cite[Thm. 3.33]{HP}. Remark that the Fr\'echet derivative $\D f [A]$ is an $\bM_n^{sa} \rightarrow \bM_n^{sa}$ map, so the formula (\ref{freform}) holds in that sense that the left hand side of (\ref{freform}) is equal to the right hand side of (\ref{freform}) restricted to $\bM_n^{sa}.$ \par

If $f$ is differentiable at $A$, then the identities
$$ 
\D f[A](B)=\kiert{\dperd{t}f(A+tB)}{t=0}
$$
and
$$ 
\tr \ler{\kiert{\dperd{t}f(A+tB)}{t=0}}=\kiert{\dperd{t} \tr f(A+tB)}{t=0}=\tr f'(A) B
$$
hold and - in particular - show that
\be \label{bregdefvar}
H_f (X,Y)=\tr \ler{f(X)-f(Y)-f'(Y)(X-Y)}. 
\ee

\begin{lemma} \label{intrepr}
If $f \in C^2 ((0, \infty)),$ the Bregman divergence admits the integral representation
\be \label{intrep_1}
H_f (X,Y)=\int_{s=0}^{1} (1-s) \tr (X-Y) \D f'[Y+s(X-Y)](X-Y) \diff{s}.
\ee 
\end{lemma}

\begin{proof}
Remark that
$$
\tr f(X)-\tr f(Y)=\int_{t=0}^{1} \dperd{t} \tr f(Y+t(X-Y)) \diff{t}=\int_{t=0}^{1} \tr f'(Y+t(X-Y))(X-Y) \diff{t}
$$
and
$$
f'(Y+t(X-Y))-f'(Y)= \int_{s=0}^{t} \dperd{s} f'(Y+s(X-Y)) \diff{s}= \int_{s=0}^{t} \D f'[Y+s(X-Y)](X-Y) \diff{s},
$$
hence
$$
H_f (X,Y)= \int_{t=0}^{1} \tr \ler{\ler{\int_{s=0}^{t} \D f'[Y+s(X-Y)](X-Y) \diff{s}} \ler{X-Y}} \diff{t}
$$
$$
=\int_{t=0}^{1} \int_{s=0}^{t}\tr (X-Y) \D f'[Y+s(X-Y)](X-Y) \diff{s}  \diff{t}
$$
$$
=\int_{s=0}^{1} (1-s) \tr (X-Y) \D f'[Y+s(X-Y)](X-Y) \diff{s}.
$$
\qed
\end{proof}


\section{A characterization of the joint convexity}

In this section we investigate the Bregman $f$-divergence from the viewpoint of joint convexity,
which is essential in the further applications. Since $f$ is convex, it is clear that the Bregman 
divergence is convex in the first variable.
For the original Bregman divergence (\ref{Breg_origin}) Bauschke and Borwein show \cite{BaBo} that $D_\phi$ 
is jointly convex - i. e. 
$$
D_\phi(t p_1+(1-t) p_2 ,t q_1+(1-t) q_2) \leq t D_\phi(p_1, q_1) +(1-t) D_{\phi}(p_2,q_2),
$$
where $p_1,p_2,q_1,q_2 \in \R^d, \, t \in [0,1]$ - if and only if the inverse of the Hessian of $\phi$ is concave in L\"owner sense. Particularly,
if $\phi$ is an $\R \supset I \rightarrow \R$ convex function, then $D_\phi$ is jointly convex if and only if $1/\phi''$ is concave. 
From this viewpoint the next characterization is rather interesting.
\begin{thm} \label{fo}
Let $f \in  C^2((0, \infty))$ be a convex function  with $f''>0$ on $\R^{++}.$ Then the following conditions are equivalent.
\ben[(A)]

\item \label{opkonk}
The map
\be \label{conc}
\bM_n^{++} \rightarrow \mc{B}\ler{\bM_n^{sa}}; \quad X \mapsto \ler{\D f'[X]}^{-1}
\ee
is operator concave.

\item \label{bregkonv} The Bregman $f$-divergence
$$
H_f: \bM_n^{++} \times \bM_n^{++} \rightarrow \R^+; \quad (X,Y) \mapsto H_f(X,Y)
$$
is jointly convex.
\een
\end{thm}
\begin{remark}
For a convex function $f \in  C^2((0, \infty))$ the property $f''>0$ is equivalent to the existence of $\ler{\D f'[X]}^{-1}$ for every $X \in \bM_n^{++}.$ On the one hand, $f''>0$ ensures that $\D f'[X] \in \mc{B}\ler{\bM_n^{sa}}$ is a positive definite and hence invertible map --- see formula (\ref{freform}). On the other hand, if $f''(\lambda)=0$ for some $\lambda>0,$ then $\D f'[\lambda I]=0 \in \mc{B}\ler{\bM_n^{sa}}.$
\end{remark}

In the recent paper \cite{trch} \emph{Tropp} and \emph{Chen} defined the \emph{Matrix Entropy Class} the following way.

\begin{defi} \label{mec}
The Matrix Entropy Class consists of the $\R^+ \rightarrow \R$ functions that are either affine or satisfy the following conditions.
\begin{itemize}
 \item $f$ is convex and $f \in C([0, \infty)) \cap C^2((0, \infty)).$
 \item For every $n \in \N,$ the map $ \bM_n^{++} \rightarrow \mc{B}\ler{\bM_n^{sa}}; \, \, X \mapsto \ler{\D f'[X]}^{-1}$ is concave with respect to the semidefinite order.
\end{itemize}
\end{defi}

By this definition, the statement of Theorem \ref{fo} is essentially the following: the set of those functions for which the corresponding Bregman divergence is jointly convex coincides with the Matrix Entropy Class defined by Tropp and Chen.

\emph{Proof of Theorem \ref{fo}:}
Let us prove the direction (\ref{opkonk}) $\Rightarrow$ (\ref{bregkonv}) first. Let $X_i$ and $Y_i$ be positive definite $n \times n$ matrices ($i \in \{1, \dots, N\}$) and let $\alpha_i$ be reals such that $\alpha_i \geq 0, \, \sum_i \alpha_i=1.$ Let us use the notations $X=\sum_i \alpha_i X_i, \, Y=\sum_i \alpha_i Y_i.$ By the operator concavity of the map $X \mapsto \ler{\D f'[X]}^{-1},$ for any $0 \leq s \leq 1$ we have
$$
\tr (X-Y) \D f'[Y+s(X-Y)](X-Y)
$$
$$
=\tr \ler{\sum_i \alpha_i (X_i-Y_i)} \ler{\ler{\D f'\left[\sum_i \alpha_i (Y_i+s(X_i-Y_i))\right]}^{-1}}^{-1} \ler{\sum_i \alpha_i (X_i-Y_i)} 
$$
$$
\leq \tr \ler{\sum_i \alpha_i (X_i-Y_i)} \ler{ \sum_i \alpha_i \ler{\D f'\left[ Y_i+s(X_i-Y_i)\right]}^{-1}}^{-1} \ler{\sum_i \alpha_i (X_i-Y_i)}.
$$
We used that taking the inverse of an operator reverses the semidefinite order.
If $\mc{H}$ is a Hilbert space, then the map
$$
\mc{H} \times \mc{B}^{++}(\mc{H}) \rightarrow \R; \quad (x, T) \mapsto \inner{x}{T^{-1}x}
$$
is convex (see \cite[Prop. 4.3]{ha06a}, which may be obtained as a consequence of \cite[Thm. 1]{lr74}). If we apply this property to the Hilbert space $\bM_n^{sa}$ with the Hilbert-Schmidt inner product we get that 
$$
\tr \ler{\sum_i \alpha_i (X_i-Y_i)} \ler{ \sum_i \alpha_i \ler{\D f'\left[Y_i+s(X_i-Y_i)\right]}^{-1}}^{-1} \ler{\sum_i \alpha_i (X_i-Y_i)} 
$$
$$
\leq \sum_i \alpha_i  \tr (X_i-Y_i) \ler{\ler{\D f'\left[(Y_i+s(X_i-Y_i))\right]}^{-1}}^{-1} (X_i-Y_i)
$$
$$
=\sum_i \alpha_i  \tr (X_i-Y_i) \D f'\left[(Y_i+s(X_i-Y_i))\right](X_i-Y_i).
$$
The result of Lemma \ref{intrepr} (eq. (\ref{intrep_1})) clearly shows that the obtained inequality
$$
\tr (X-Y) \D f'[Y+s(X-Y)](X-Y)\leq \sum_i \alpha_i  \tr (X_i-Y_i) \D f'\left[(Y_i+s(X_i-Y_i))\right](X_i-Y_i)
$$
implies the joint convexity of the Bregman divergence.
\par
The proof of (\ref{bregkonv}) $\Rightarrow$ (\ref{opkonk}) is the following.
The conditon (\ref{bregkonv}) means that if $A_i \in \bM_n^{++}, \, B_i \in \bM_n^{sa}$ and $\alpha_i \geq 0, \, \sum_i \alpha_i=1,$ then
\be \label{aibi_1}
H_f\ler{\sum_i \alpha_i (A_i+ \eps B_i), \sum_i \alpha_i A_i} \leq \sum_i \alpha_i H_f\ler{A_i+ \eps B_i, A_i},
\ee
where $\eps< \eps_0$ for some $\eps_0>0.$ By the integral representation (\ref{intrep_1}), the right hand side of (\ref{aibi_1}) can be written as
$$
\sum_i \alpha_i H_f\ler{A_i+ \eps B_i, A_i}=\sum_i \alpha_i \int_{s=0}^{1} (1-s) \tr \eps B_i \D f'[A_i+s \eps B_i](\eps B_i) \diff{s}
$$

$$
=\eps^2 \int_{s=0}^{1} (1-s) \sum_i \alpha_i \tr B_i \D f'[A_i+s \eps B_i](B_i) \diff{s}.
$$
Similarly, the left hand side is
$$
 H_f\ler{ \sum_i \alpha_i \ler{A_i+ \eps B_i},\sum_i \alpha_i A_i}
$$
$$
=\eps^2 \int_{s=0}^{1} (1-s) \tr \ler{\sum_i \alpha_i B_i} \D f'\left[\sum_i \alpha_i (A_i+s \eps B_i)\right] \ler{\sum_i \alpha_i B_i} \diff{s}.
$$
The assumption $f\in C^2((0, \infty))$ ensures that the map  $\D f': \ \bM_n^{++} \rightarrow \mc{B}\ler{\bM_n^{sa}}$ is continuous. Therefore, $\lim_{\eps \to 0} \D f'[A_i+s \eps B_i]=\D f'[A_i]$ etc.
After division by $\eps^2$ and taking the limit $\eps \to 0$ we obtain from (\ref{aibi_1}) that
$$
\tr \ler{\sum_i \alpha_i B_i} \D f'\left[\sum_i \alpha_i A_i\right] \ler{\sum_i \alpha_i B_i} \leq \sum_i \alpha_i \tr B_i \D f'[A_i](B_i),
$$
that is, the map 
\be \label{jc}
\bM_n^{++} \times \bM_n^{sa} \ni (A,B) \mapsto \tr B \D f'[A](B)
\ee
is jointly convex. This is sufficient to show the opearator concavity of the map $X \mapsto \ler{\D f'[X]}^{-1}$ by the followings. Let
$A_i \in \bM_n^{++}$ ($i \in \{1, \dots, N\}$) and $\alpha_i \geq 0, \, \sum_i \alpha_i=1.$ Let us use the short notation $T_i=\D f'[A_i].$ For any $C \in \bM_n^{sa}$ we can define
$$
B_i:=\ler{\D f'[A_i]}^{-1}\circ\ler{\sum_j \alpha_j \ler{\D f'[A_j]}^{-1}}^{-1}(C) \equiv T_i^{-1}\circ \ler{\sum_j \alpha_j T_j^{-1}}^{-1}(C). 
$$
Observe that by this definition $\sum_i \alpha_i B_i=C.$ On the one hand,
$$
\sum_i \alpha_i \tr B_i \D f'[A_i](B_i)=\sum_i \alpha_i \tr B_i T_i(B_i)
$$
$$
=\sum_i \alpha_i \tr \ler{ T_i^{-1}\circ \ler{\sum_j \alpha_j T_j^{-1}}^{-1}(C) \cdot T_i \circ T_i^{-1} \circ  \ler{\sum_j \alpha_j T_j^{-1}}^{-1}(C)}
$$
$$
= \tr \ler{ \ler{\sum_i \alpha_i T_i^{-1}}\circ \ler{\sum_j \alpha_j T_j^{-1}}^{-1}(C) \cdot \ler{\sum_j \alpha_j T_j^{-1}}^{-1}(C)}
$$
$$
= \tr C \cdot \ler{\sum_i \alpha_i T_i^{-1}}^{-1}(C)= \tr C \ler{\sum_i \alpha_i \ler{\D f'[A_i]}^{-1}}^{-1}(C).
$$
On the other hand,
$$
\tr \ler{\sum_i \alpha_i B_i} \D f'\left[\sum_i \alpha_i A_i\right] \ler{\sum_i \alpha_i B_i}=\tr C \D f'\left[\sum_i \alpha_i A_i\right](C)
$$
By the joint convexity of (\ref{jc}),
$$
\tr C \D f'\left[\sum_i \alpha_i A_i\right](C) \leq \tr C \ler{\sum_i \alpha_i \ler{\D f'[A_i]}^{-1}}^{-1}(C)
$$
holds, and $C$ was an arbitrary element of $\bM_n^{sa},$ hence the operator inequality
$$
\D f'\left[\sum_i \alpha_i A_i\right] \leq \ler{\sum_i \alpha_i \ler{\D f'[A_i]}^{-1}}^{-1}
$$
holds, which is equivalent to
$$
\ler{\D f'\left[\sum_i \alpha_i A_i\right]}^{-1} \geq \sum_i \alpha_i \ler{\D f'[A_i]}^{-1}.
$$
This is the desired concavity property.
\qed

\subsection{An extension of the Bregman divergence to singular matrices}
In quantum information theory, the singular density matrices play a central role, therefore, we would like to extend the Bregman $f$-divergences from $\bM_n^{++} \times \bM_n^{++}$ to $\bM_n^{+} \times \bM_n^{+}$.
It is a natural idea to define the Bregman divergence of the positive semidefinite matrices $X$ and $Y$ as follows:

\be \label{bregeps}
H_f(X,Y):=\lim_{\eps \to 0}H_f\ler{X+\eps I,Y+\eps I}.
\ee
With the formula (\ref{bregdefvar}) in hand, easy computation shows that if $X$ and $Y$ admit the spectral decompositions $X=\sum_{j=1}^{n} \lambda_j \ket{\varphi_j} \bra{\varphi_j}$ and $Y=\sum_{k=1}^{n} \mu_k \ket{\psi_k} \bra{\psi_k},$ then
\be \label{kiir}
H_f(X+\eps I,Y+\eps I)= \sum_{j,k=1}^{n} \abs{\braket{\varphi_j}{\psi_k}}^2 \ler{f(\lambda_j+\eps)-f(\mu_k+\eps) -f'(\mu_k+\eps)(\lambda_j-\mu_k)}.
\ee
Assume that $f \in C^0([0,\infty)) \cap C^1((0,\infty)),$ that is, $\lim_{x \to 0} f(x) \in \R.$
The convexity of $f$ gives that $f'$ is monotone increasing, hence $\lim_{\eps \to 0} f'(\eps) \in \R$ or $\lim_{\eps \to 0} f'(\eps)=-\infty.$

Clearly, if $\lim_{\eps \to 0} f'(\eps) \in \R,$ then the limit of (\ref{kiir}) is a real number. If $\mathrm{ker} (Y) \subseteq \mathrm{ker} (X),$ then $\lambda_j=0$ whenever $\mu_k=0$ and $\braket{\varphi_j}{\psi_k} \neq 0,$ hence the limit is finite in this case, as well. If $\mathrm{ker} (Y) \nsubseteq \mathrm{ker} (X)$ and $\lim_{\eps \to 0} f'(\eps)=-\infty,$ then the limit is $+\infty.$
\par
So we conclude that if $f$ is continuous at $0,$ then (\ref{bregeps}) is well-defined and takes values in $\R^+ \cup \{+ \infty \},$ that is, the Bregman $f$-divergences can be extended to $\bM_n^{+} \times \bM_n^{+}$ by continuity. \par

If $H_f(\cdot, \cdot)$ defined by a convex function $f \in  C^2((0, \infty))$ is jointly convex on $\bM_n^{++} \times \bM_n^{++},$ then (assuming in addition that $f \in  C^0([(0, \infty))$) $H_f(\cdot, \cdot)$ is jointly convex on $\bM_n^{+} \times \bM_n^{+}.$ Indeed, for $X_k, Y_k \in \bM_n^{+}, c_k\geq 0, \sum_k c_k=1$ we have
$$
H_f \ler{\sum_k c_k X_k, \sum_k c_k Y_k}= \lim_{\eps \to 0} H_f \ler{\sum_k c_k X_k+\eps I, \sum_k c_k Y_k+\eps I}
$$
$$
=\lim_{\eps \to 0} H_f \ler{\sum_k c_k (X_k+\eps I), \sum_k c_k (Y_k+\eps I)} \leq
\lim_{\eps \to 0} \sum_k c_k H_f \ler{X_k+\eps I, Y_k+\eps I}
$$
$$
=\sum_k c_k \lim_{\eps \to 0} H_f \ler{X_k+\eps I, Y_k+\eps I}=\sum_k c_k H_f \ler{X_k, Y_k}.
$$
Therefore, we can reformulate the main condition with a bit different conditions.
\begin{thm} \label{fo2}
Let $f \in  C^0([0,\infty)) \cap C^2((0, \infty))$ be a convex function with $f''>0$ on $\R^{++}.$ Then the following conditions are equivalent.
\ben[(i)]
\item \label{opkonk2}
The map
$$
\bM_n^{++} \rightarrow \mc{B}\ler{\bM_n^{sa}}; \quad X \mapsto \ler{\D f'[X]}^{-1}
$$
is operator concave.

\item \label{bregkonv2} The Bregman $f$-divergence
$$
H_f: \bM_n^{+} \times \bM_n^{+} \rightarrow \R^+ \cup \{+ \infty \}; \quad (X,Y) \mapsto H_f(X,Y)
$$
is jointly convex.
\een
\end{thm}

\subsection{A different condition and alternative proofs}
In a recent preprint Hansen and Zhang investigated the connections between the condition (\ref{opkonk}) in Theorem \ref{fo} and the property that $f''$ is operator convex and numerically non-increasing. In an earlier version ot their paper, these conditions were claimed to be equivalent \cite[Thm 1.2]{hazh}. Later the proof turned out to be incomplete. In the current version it is proved that if $f''$ is operator convex and numerically non-increasing, then the condition (\ref{opkonk}) in Theorem \ref{fo} is satisfied \cite[Thm 1.3]{hazh_v}.
\par
Now we give a direct proof of the fact that the operator convexity (and the non-increasing property) of $f''$ is sufficient to deduce the joint convexity of the Bregman $f$-divergence.


\begin{lemma} \label{diffrepr}
Set $f \in C^1 ((0, \infty))$ and $A \in \bM_n^{++}.$ Then the Fr\'echet derivative of the standard matrix function $f$ is
$$
\D f [A]=
\int_{0}^{1} f'\ler{t L_A + (1-t)R_A} \mathrm{d}t,
$$
where $L_A$ ($R_A$) denotes the left (right) multiplication by $A:$
$$
L_A: \bM_n \rightarrow \bM_n, \ X \mapsto L_A(X):=AX, \qquad R_A: \bM_n \rightarrow \bM_n, \ X \mapsto R_A(X):=XA.
$$
\end{lemma}
\begin{proof}
Let us use the notations $A=\sum_j \lambda_j \ket{\varphi_j}\bra{\varphi_j}$ and $E_{ij}=\ket{\varphi_i}\bra{\varphi_j}$ again. It is easy to check that
$$
L_A(E_{ij})=\lambda_i E_{ij}, \ R_A(E_{ij})= \lambda_j E_{ij},
$$
hence with $P_{ij}:=\ket{E_{ij}}\bra{E_{ij}}$ we have
$$
L_A= \sum_{i,j} \lambda_i P_{ij}, \ R_A= \sum_{i,j} \lambda_j P_{ij}.
$$
Therefore
$$
f'\ler{t L_A + (1-t)R_A}= \sum_{i,j} f'(t \lambda_i +(1-t)\lambda_j) P_{ij},
$$
and
$$
\int_{0}^{1} f'\ler{t L_A + (1-t)R_A} \mathrm{d}t= \sum_{i,j} \int_{0}^{1} f'\ler{\lambda_j +t(\lambda_i-\lambda_j)} \mathrm{d}t \ket{E_{ij}}\bra{E_{ij}},
$$
and this exactly the formula that appeared in (\ref{freform}).
\qed
\end{proof}
Lemma \ref{intrepr} and Lemma \ref{diffrepr} have an immediate consequence.

\begin{cor} \label{kov1}
For $f \in C^2 ((0, \infty)),$ the Bregman divergence can be written as
\be \label{jorepr}
H_f (X,Y)=\int_{s=0}^{1} \int_{t=0}^{1}  (1-s)\tr (X-Y)  f''\ler{ t L_{Y+s(X-Y)} + (1-t)R_{Y+s(X-Y)} } (X-Y) \diff{t} \diff{s}.
\ee
\end{cor}

\begin{thm} \label{mellek}
Let $f \in  C^2((0, \infty))$ be a convex function. If $f''$ is operator convex and numerically non-increasing, then the Bregman $f$-divergence
$$
H_f: \bM_n^{++} \times \bM_n^{++} \rightarrow \R^+; \quad (X,Y) \mapsto H_f(X,Y)
$$
is jointly convex.
\end{thm}

\emph{First proof of Theorem \ref{mellek}:}

On a Hilbert space $\mc{H}$ the map
$$\
\mc{B}(\mc{H})^{++} \times \mc{H} \rightarrow \R: \ (A, \xi) \mapsto \inner{\xi}{\varphi(A) (\xi)}
$$
is jointly convex if $\varphi: (0, \infty) \rightarrow \R$ is operator convex and numerically non-increasing. This fact relies on the joint convexity of the map $(A, \xi) \mapsto \inner{\xi}{A^{-1} \xi}$ --- which is stated in \cite[Prop. 4.3]{ha06a} and may be derived from \cite[Thm. 1]{lr74} --- and on an integral representation of the functions $\varphi$ with the above property. This representation will be discussed in the second proof of this theorem.
The maps
$$
(X,Y) \mapsto t L_{Y+s(X-Y)} + (1-t)R_{Y+s(X-Y)}
$$
and  $(X,Y) \mapsto X-Y$ are affine, and with the Hilbert-Schmidt inner product (\ref{jorepr}) can be written as
\be \label{joreprHS}
H_f (X,Y)=\int_{s=0}^{1} \int_{t=0}^{1}  (1-s)\inner {X-Y}  {f''\ler{ t L_{Y+s(X-Y)} + (1-t)R_{Y+s(X-Y)} } (X-Y)} \diff{t} \diff{s},
\ee
hence $H_f$ is jointly convex if $f''$ is operator convex and non-increasing.
\qed
\par

We may provide another proof of this theorem.

\begin{prop}
Let $\mc{F}(A,B)$ denote the set of all $A \rightarrow B$ functions. The map
\be \label{breg_alg}
H: C^2((0, \infty)) \rightarrow \mc{F}\ler{\bM_n^{++} \times \bM_n^{++},\R}, \, f \mapsto H_f (\cdot, \cdot)
\ee
is linear, and the kernel is the subspace of affine functions, that is, $$\mathrm{Ker}(H)=\{x\mapsto ax+b \, | \, a,b \in \R\}.$$
\end{prop}
\begin{proof}
The linearity is obvious, and from the integral formula (\ref{joreprHS}) it is easy to see that the kernel of $H$ is equal to the kernel of the operator $\ddperd{x^2}: C^2((0,\infty)) \rightarrow C((0,\infty)).$
\qed 
\end{proof}
Therefore, if $f$ can be written as $f=\sum_{j=1}^k f_j,$ where the $f_j$'s define jointly convex Bregman divergence, then $H_f(\cdot, \cdot)$ is jointly convex. The affine part of a function can be omitted. \par

\emph{Second proof of Theorem \ref{mellek}:}
If $f \in C^2 ((0, \infty))$ is convex function, then $f''$ is numerically non-increasing and operator convex if and only if
\be \label{int_form_0}
f''(x)= \gamma + \int_{0}^{\infty} \frac{1}{\lambda + x} \diff{\mu (\lambda)},
\ee
where $\gamma \geq 0$ and $\mu$ is a nonnegative measure on $[0,\infty)$ such that
$$
\int_{0}^{\infty} \frac{1}{1+\lambda} \diff{\mu (\lambda)} < \infty.
$$
This fact is stated in this form in \cite[Thm. 3.1]{ah11}. Now we intend to use the 'only if' direction, hence we outline the key steps of the proof of Ando and Hiai. A rather complex argument shows that a numerically non-increasing and operator convex function is operator monotone decreasing. At this point, Ando and Hiai provide a slightly simplificated version of the original argument of Hansen \cite{ha06b} to verify the integral representation. This is the following. If $f''$ is operator monotone decreasing then $f''\ler{\frac{1}{x}}$ is operator monotone, hence by \cite[pp. 144-145]{bhatia} it has the form
\be \label{monrepr}
f''\ler{\frac{1}{x}}= \alpha+ \beta x + \int_{0}^{\infty} \frac{(\lambda+1)x}{\lambda+x}\diff{\nu(\lambda)}
\ee
where $\alpha, \beta \geq 0$ and $\nu$ is a nonnegative finite measure on $(0,\infty).$ Let $\diff{\tilde\mu(\lambda)}:= \diff{\nu(\frac{1}{\lambda})}$ on $(0,\infty)$ and $\tilde \mu (\{0\}):=\beta.$ Finally, $\diff{\mu (\lambda)}:=(\lambda+1)\diff{\tilde \mu (\lambda)}.$
By these operations, the representation (\ref{monrepr}) is transformed to (\ref{int_form_0}). \par

Integrating (\ref{int_form_0}) two times with respect to $x$ we get
\be \label{int_form_1}
f(x)=\alpha + \beta x +\frac{\gamma}{2}x^2  + \int_{0}^{\infty} \ler{(\lambda + x)\ler{\log{(\lambda + {x})}-\log{(\lambda+1)}}-(x-1)} \diff{\mu (\lambda)},
\ee
where $\alpha, \beta \in \R.$
One can see that $f(x)$ is the sum of the affine part
$$ a(x)=\alpha + \beta x -\int_{0}^{\infty} \ler{\log{(\lambda+1)}(\lambda + x)+(x-1)}\diff{\mu (\lambda)},$$
the quadratic part $q(x)=\frac{\gamma}{2}x^2$ and the ``entropic" part
$$ e(x)= \int_{0}^{\infty} (\lambda + x) \log{(\lambda + x)} \diff{\mu (\lambda)}.$$

The quadratic part $H_q(X,Y)=\frac{\gamma}{2}\tr (X-Y)^2$ is clearly jointly convex. By the result of \cite{Linb74}, the same statement holds for the Bregman divergence induced by the standard entropy function $\varphi_0(x)=x \log{x},$
$$
H_{\varphi_0}(X,Y)=\tr \ler{X\ler{\log X - \log Y}-(X-Y)}.
$$
On the other hand, one can check that the Bregman divergence induced by the shifted entropy function $\varphi_\lambda(x)=(x+\lambda) \log{(x+\lambda)}$ can be expressed as
\be \label{shift}
H_{\varphi_\lambda}(X,Y)=H_{\varphi_0}(X+\lambda I, Y+\lambda I).
\ee
On the whole, if $f''$ is numerically decreasing and operator convex, then the Bregman divergence $H_f$ can be written as $H_f=H_q+H_a+H_e,$ where $H_a=0,$ $H_q$ is obviously jointly convex and
\be \label{utcso}
H_e(X,Y)=\int_{0}^{\infty} H_{\varphi_\lambda}(X,Y)\diff{\mu (\lambda)}=\int_{0}^{\infty} H_{\varphi_0}(X+\lambda I, Y+\lambda I)\diff{\mu (\lambda)}.
\ee
The map $(X,Y) \mapsto (X+\lambda I, Y+\lambda I)$ is affine, hence (\ref{utcso}) is jointly convex, and this completes the proof.
\qed



\section{An application - the Tsallis entropy}

For any real $q,$ one can define the deformed logarithm (or $q$-logarithm) function $\ln_q: \R^{++} \rightarrow \R$ by
$$
\ln_{q}\,x =\int_{1}^{x} t^{q-2} \mathrm{d}t=\begin{cases} \frac{x^{q-1}-1}{q-1} &\mbox{if }
q \neq 1 \, , \\ \ln\, x & \mbox{if } q=1\, . \end{cases}
$$
If we define
$
f_q(x):=x \ln_q (x)
$
then the \emph{Tsallis entropy} \cite{A-D, Dar} of a density matrix $\rho$ (i.e. $\rho\ge 0$ and 
$\tr\rho=1$) is given by
$$
S_q(\rho)=\tr f_q (\rho).
$$
Note that if $q>0,$ then $\lim_{x \to 0} f_q (x)=0,$ hence $f_q$ can be extended by continuity, thus the Tsallis entropy is well-defined for singular densities, as well.
By the result of Tropp and Chen, $f_q$ belongs to the Matrix Entropy Class for $1 \leq  q \leq 2$ \cite[Thm. 2.3]{trch}.
Therefore, by Theorem \ref{fo}, $H_{f_q}(\cdot,\cdot)$ is jointly convex. (Alternatively, we may use the well-known operator convexity of the function $x \mapsto x^{-r} \,\, (0\leq r \leq 1)$ and refer to Theorem \ref{mellek}.)

One can compute that for $q \neq 1$ we have
$$
H_{f_q}(A,B)=\tr B^q+\frac{1}{q-1}\ler{\tr A^q- q \tr A B^{q-1}}.
$$
The Bregman divergence is unitary invariant, that is, $H_f(UAU^*, UBU^*)=H_f(A,B)$ for all unitary matrices $U$. If $X \in \bM_m \otimes \bM_n$ then there are some unitaries such that
$$
X_1 \otimes \frac{1}{n} I_2 = \sum_{k=1}^{n^2} \frac{1}{n^2} U_k X U_k^*,
$$
where $X_1=\tr_2 X$ and $I_2$ is the identity in $\bM_n$ (see e. g. \cite{BP, Furu-fund}), hence from the joint convexity it follows that the Bregman divergence is monotone in the following sense:
\be \label{mon}
H_f \ler{X_1 \otimes \frac{1}{n} I_2, Y_1 \otimes \frac{1}{n} I_2} \leq H_f \ler{X, Y}
\ee
if $f$ satisfies the condition (\ref{opkonk}) in Theorem \ref{fo}.
Let us apply (\ref{mon}) to $f_q$ with $1 < q \leq 2$ and
\be \label{val}
X=\rho_{123} \in \mc{B}^{+}\ler{\mc{H}_1 \otimes \mc{H}_2 \otimes \mc{H}_3}, \, Y=\frac{1}{d_1} I_1 \otimes \rho_{23},
\ee
where $\mc{H}_i$ is a finite dimensional Hilbert space ($i \in \{1,2,3\}$), $d_i=\mathrm{dim} \mc{H}_i$ and $\rho_{23}=\tr_1 \rho_{123}.$ The idea of this choice comes from the tutorial \cite{NP05}. 
With this choice we get
\be \label{monkonkr} 
H_{f_q} \ler{\rho_{12} \otimes \frac{1}{d_3} I_3, \frac{1}{d_1} I_1 \otimes \rho_2 \otimes \frac{1}{d_3} I_3} \leq H_{f_q} \ler{\rho_{123}, \frac{1}{d_1} I_1 \otimes \rho_{23}}.
\ee
Straightforward computations show that the left hand side of (\ref{monkonkr}) equals to
$$
\frac{1}{q-1}\ler{d_3^{1-q} \tr \rho_{12}^q-(d_1 d_3)^{1-q} \tr \rho_2^q}
$$
and the right hand side is
$$
\frac{1}{q-1} \ler{ \tr \rho_{123}^q-d_1^{1-q} \tr \rho_{23}^q}.
$$
The result of this computation can be summarized as follows.
\begin{thm} \label{tsaqssa}
If $\mc{H}_i$ is a finite dimensional Hilbert space for any $i \in \{1,2,3\},$ $d_i=\mathrm{dim} \mc{H}_i,$ $1 \leq q \leq 2,$ then for any $\rho_{123} \in \mc{B}^{+}\ler{\mc{H}_1 \otimes \mc{H}_2 \otimes \mc{H}_3}$
the inequality
\be \label{all}
d_3^{1-q} \tr \rho_{12}^q + d_1^{1-q} \tr \rho_{23}^q \leq \tr \rho_{123}^q +(d_1 d_3)^{1-q} \tr \rho_2^q.
\ee
holds, where notations like $\rho_{12}$ denote the appropriate reduced matrices.
\end{thm}

The fact that the strong subadditivity of the Tsallis entropy
$$ 
\tr \rho_{12}^q + \tr \rho_{23}^q \leq \tr \rho_{123}^q +\tr \rho_2^q
$$
does not hold in general \cite{PV} (but holds for classical probability distributions \cite{Furu}) makes Theorem \ref{tsaqssa} remarkable. Furthermore, one can not state more that (\ref{all}), the inequality is sharp. The density matrix
\be \label{exa}
\rho_{123}=\left[ \begin {array}{cccccccc}  0 &0&0&0&0&0&0&0\\
\noalign{\medskip}0& 0 &0&0&0&0&0&0\\\noalign{\medskip}0&0& \frac{1}{4}&0& \frac{1}{4}&0&0&0\\
\noalign{\medskip}0&0&0& \frac{1}{4}&0& \frac{1}{4}&0&0\\\noalign{\medskip}0&0&
\frac{1}{4}&0& \frac{1}{4}&0&0&0\\
\noalign{\medskip}0&0&0& \frac{1}{4}&0& \frac{1}{4}&0&0\\\noalign{\medskip}0&0&0&0&0&0&0&0\\
\noalign{\medskip}0&0&0&0&0&0&0& 0
\end {array} \right] \in \mc{B}\ler{\C^2 \otimes \C^2 \otimes \C^2}
\ee
has the reduced densities
$$
\rho_{12}=\left[ \begin {array}{cccc}  0&0&0&0\\\noalign{\medskip}0& \frac{1}{2}& \frac{1}{2}&0\\
\noalign{\medskip}0& \frac{1}{2}& \frac{1}{2}&0\\\noalign{\medskip}0&0&0&0\end {array} \right],
\quad
\rho_{23}=\left[ \begin {array}{cccc}  \frac{1}{4}& 0& 0& 0\\\noalign{\medskip} 0& \frac{1}{4}& 0& 0\\
\noalign{\medskip} 0& 0& \frac{1}{4}& 0\\\noalign{\medskip} 0& 0& 0& \frac{1}{4}\end {array} \right],
\quad
\rho_2=\left[ \begin {array}{cc}  \frac{1}{2}& 0\\\noalign{\medskip} 0& \frac{1}{2}\end {array} \right],
$$
hence
$$
d_3^{1-q} \tr \rho_{12}^q + d_1^{1-q} \tr \rho_{23}^q =2^{1-q}+2^{1-q} 4^{1-q}= \tr \rho_{123}^q +(d_1 d_3)^{1-q} \tr \rho_2^q.
$$
This example appeared in \cite{PV} to demonstrate that the Tsallis entropy is not strongly subadditive.
\par
Note that (\ref{all}) is equivalent to
$$
(d_1 d_2 d_3)^{1-q} S_q(\rho_{123})+d_2^{1-q} S_q(\rho_{2})+\frac{(d_1 d_2 d_3)^{1-q}-1}{q-1}+\frac{d_2^{1-q}-1}{q-1}
$$
$$
\leq
(d_2 d_3)^{1-q} S_q(\rho_{23})+(d_1 d_2)^{1-q} S_q(\rho_{12})+\frac{(d_2 d_3)^{1-q}-1}{q-1}+\frac{(d_1 d_2)^{1-q}-1}{q-1}, 
$$
which gives the strong subadditivity
$$
S_q(\rho_{123})+S_q(\rho_{2}) \leq S(\rho_{23})+S(\rho_{12})
$$
of the von Neumann entropy, if we take the limit $q \to 1.$ \par
Another inequality can be derived if we consider (\ref{mon}) with $X=\rho_{12}$ and $Y=\frac{1}{d_1} I_1 \otimes \rho_2.$ With this choice
$$
H_{f_q} \ler{\rho_{1} \otimes \frac{1}{d_2} I_2, \frac{1}{d_1} I_1 \otimes \frac{1}{d_2} I_2}=d_2^{1-q} \tr \rho_1^q -(d_1 d_2)^{1-q}
$$
and
$$
H_{f_q} \ler{\rho_{12}, \frac{1}{d_1} I_1 \otimes \rho_{2}}=\tr \rho_{12}^q -d_1^{1-q} \tr \rho_{2}^q,
$$
hence the monotonicity (\ref{mon}) gives that
\be \label{ketto}
d_1^{q-1} \tr \rho_1^q + d_2^{q-1} \tr \rho_2^q \leq (d_1 d_2)^{q-1} \tr \rho_{12}^q +1.
\ee
Note that (\ref{ketto}) is a special case of (\ref{all}) with the trivial subsystem $\mc{H}_2=\C$ in (\ref{val}).

\section{The relation of joint convexity and monotonicity under stochastic maps}
For \emph{homogeneous} relative entropy-type maps, the joint convexity and the monotonicity under stochastic maps is equivalent \cite[remarks after Def. 2.3]{lr99}. However, the Bregman divergence does not need to be homegeneous. For example,
\be \label{qhom}
H_{f_q}( \lambda A, \lambda B)= \lambda^q H_{f_q}(A,B)
\ee
for $0 < \lambda$ (and any positive $q$).
\par
We show that a jointly convex Bregman divergence is not monotone in general. In order to see this surprising fact, we create an example that shows that a family of (jointly convex) Bregman divergences increases under the partial trace, which is a very important stochastic (that is, completely positive trace preserving - CPTP) map \cite{Carlen, lr99}. \par
Easy computations show that for $A, B \in \bM_m^{+}$
\be \label{tobb0}
H_{f} \ler{A \otimes \frac{1}{n} I_2, B \otimes \frac{1}{n} I_2} = n H_{f} \ler{\frac{A}{n}, \frac{B}{n}},
\ee
where $I_2$ is the identity in $\bM_n.$ Recall that the density matrix (\ref{exa}) saturates the inequality (\ref{monkonkr}), that is,

\be \label{mksat} 
H_{f_q} \ler{\rho_{12} \otimes \frac{1}{2} I, \frac{1}{2} I \otimes \rho_2 \otimes \frac{1}{2} I} = H_{f_q} \ler{\rho_{123}, \frac{1}{2} I \otimes \rho_{23}}.
\ee
On the other hand, by (\ref{qhom}) and (\ref{tobb0}),
\be \label{tobb}
H_{f_q} \ler{\rho_{12} \otimes \frac{1}{2} I, \frac{1}{2} I \otimes \rho_2 \otimes \frac{1}{2} I} = 2^{1-q} H_{f_q} \ler{\rho_{12}, \frac{1}{2} I \otimes \rho_2},
\ee
which means that
$$
H_{f_q} \ler{\rho_{12}, \frac{1}{2} I \otimes \rho_2}>H_{f_q} \ler{\rho_{123}, \frac{1}{2} I \otimes \rho_{23}},
$$
so the monotonicity under partial trace fails. \par
This means that the joint convexity does not imply monotonicity, but the converse is true.
We summarize the results in the next theorem.

\begin{thm} \label{jcmon}
Every monotone Bregman divergence is jointly convex. However, there are jointy convex Bregman divergences which are not monotone under stochastic maps. On the other hand, joint convexity implies the monotonicity under the stochastic maps of the form
\be \label{cpspec}
A \mapsto \sum_k c_k U_k A U_k^*,
\ee
where the $U_k$'s are unitaries and $c_k \geq 0,$ $\sum_k c_k =1.$
 
\end{thm}

\begin{proof}
Set $X_1, X_2, Y_1, Y_2 \in \bM_n^{+}$ and let the block matrices $X, Y$ and $U \in \bM_{2n}$ be defined by
$$
X=\left[\ba{cc} X_1 & 0 \\ 0 & X_2 \ea \right], \, Y=\left[\ba{cc} Y_1 & 0 \\ 0 & Y_2 \ea \right], \ U=\left[\ba{cc} 0 & I \\ I & 0 \ea \right],
$$
where $I \in \bM_n$ is the identity matrix.
The map
$$
\mc{E}: \bM_{2n} \rightarrow \bM_{2n}, \, X \mapsto \mc{E}(X):=\frac{1}{2}X+\frac{1}{2}UXU^*
$$
is clearly stochastic, and
$$
\mc{E}(X)= \frac{1}{2}\left[\ba{cc} X_1+X_2 & 0 \\ 0 & X_1+X_2 \ea \right], \, \mc{E}(Y)= \frac{1}{2}\left[\ba{cc} Y_1+Y_2 & 0 \\ 0 & Y_1+Y_2 \ea \right].
$$
The Bregman divergence of block-diagonal matrices is the sum of the Bregman divergence of the blocks, hence the monotonicity condition
$$
H_f\ler{\mc{E}(X), \mc{E}(Y)} \leq H_f(X,Y)
$$
means that
$$
2 H_f\ler{\frac{1}{2} (X_1+X_2), \frac{1}{2} (Y_1+Y_2)} \leq H_f(X_1, Y_1)+H_f(X_2, Y_2),
$$
which is the midpoint convexity of $H_f(\cdot, \cdot).$ The Bregman $f$-divergence is continuous (by the assumption $f \in C^{1}((0, \infty))$), hence midpoint convexity implies convexity.
\par
We have shown in this section that for $f_q(x)=\frac{x^q-q}{q-1}$ the corresponding Bregman divergence $H_{f_q} (\cdot, \cdot)$ is jointly convex but it is not monotone under stochastic maps ($1 < q \leq 2$).
\par
In order to check the last statement of the theorem, suppose that $H_f(\cdot, \cdot)$ is jointly convex. If a map $\mc{E}: \bM_n \rightarrow \bM_n$ has the form (\ref{cpspec}), then by the unitary invariance of the Bregman divergence,
$$
H_f \ler{\mc{E}(X),\mc{E}(Y)}=H_f \ler{\sum_k c_k U_k X U_k^*,\sum_k c_k U_k Y U_k^*} \leq \sum_k c_k H_f \ler{ U_k X U_k^*,U_k Y U_k^*}
$$
$$
=\sum_k c_k H_f \ler{X,Y}=H_f \ler{X,Y}.
$$
\qed
\end{proof}
\subsection{A possible consequence of the joint convexity}
Provided that the characterization of the Matrix Entropy Class by Hansen and Zhang (\cite[Thm. 1.2]{hazh}) is true, we can deduce that the joint convexity of the Bregman $f$-divergence implies another monotonicity property.
By the result of Theorem \ref{fo} and by the integral representation (\ref{int_form_0}), if $H_f(\cdot, \cdot)$ is jointly convex, then
$$
f'(x)= \delta+\gamma x+ \int_{0}^{\infty} \log{\ler{\lambda + x}} \diff{\mu (\lambda)},
$$
where $\delta \in \R, \, \gamma \geq 0$ and $\mu$ is the same measure as in (\ref{int_form_0}).
The function $x \mapsto \log{\ler{\lambda + x}}$ is operator monotone on $(0,\infty)$ for any nonnegative $\lambda$ (see e. g. \cite{Carlen, HP}), hence so is $f'$ (the integration keeps the monotonicity).

\par
In the recent paper \cite{ls14} Lewin and Sabin showed that the operator monotonicity of $f'$ is equivalent to the following monotonicity property: for any $A,B \in \bM_n^+$ and $X \in \bM_{n \times k}$ with $X^{*}X \leq I \in \bM_n$ we have
\be \label{mon_var}
H_f(XAX^{*}, XBX^{*}) \leq H_f(A,B).
\ee
Thus we deduced that a jointly convex Bregman divergence is monotone in the sense of (\ref{mon_var}).

{\bf Acknowledgments.} This work was partially supported by the Hungarian Research Grant
OTKA K104206. The authors would like to thank the anonymous referee for his/her constructive remarks and suggestions and Prof. D\'enes Petz for great conversations. DV is grateful to Anna Jencova for illuminating discussions.



\begin{thebibliography}{9}

  \bibitem[AD75]{A-D}
  J. Acz\'el and Z. Dar\'oczy, {\em{On Measures of Information and Their Characterizations}}, Academic Press, San Diego, 1975.
  
  \bibitem[AH11]{ah11} T. Ando and F. Hiai, Operator log-convex functions and operator means, Mathematische Annalen {\bf 350(3)}(2011), 611-630. 
  
  \bibitem[BA05]{ba05} A. Banerjee \emph{et al.}, Clustering with Bregman Divergences, J. Mach. Learn. Res. {\bf 6}(2005), 1705-1749.
  
  \bibitem[BB01]{BaBo} H. Bauschke and J. Borwein, Joint and separate convexity of the Bregman distance, Inherently Parallel Algorithms in Feasibility and Optimization and their Applications (Haifa 2000), D. Butnariu, Y. Censor, S. Reich (editors), Elsevier, pp. 23-36, 2001.
  
  \bibitem[BP13]{BP} \'A.  Besenyei and D. Petz, Partial subadditivity of entropies, Linear Algebra and its Applications {\bf 439}(2013), 3297\,-\,3305.
  \bibitem[BH96]{bhatia} R. Bhatia, \emph{Matrix Analysis,} Springer-Verlag, New York, 1996.
  
  \bibitem[BR67]{Breg} L. M. Bregman, The relaxation method of finding the common points of convex sets and its application to the solution of problems in convex programming, USSR Computational Mathematics and Mathematical Physics {\bf 7(3)}(1967), 200-217.
    
 \bibitem[CA10]{Carlen} E. Carlen, Trace Inequalities and Quantum Entropy: An Introductory Course, Contemp. Math. {\bf 529}(2010), 73-140.
 
 \bibitem[CT14]{trch} R. Y. Chen and J. A. Tropp, Subadditivity of matrix $\varphi$-entropy and concentration of random matrices, Electron. J. Probab. {\bf 19}(2014), 1-30.
 
 \bibitem[DA70]{Dar}
  Z. Dar\'oczi, General information functions, Information and Control {\bf 16}(1970), 36\,-\,51.
  
  \bibitem[FU04]{Furu-fund} S. Furuichi, K. Yanagi and K. Kuriyama, Fundamental properties of Tsallis relative entropy, J. Math.Phys. {\bf 45}(2004), 4868-4877.
  
  \bibitem[FU06]{Furu} S. Furuichi, Information theoretical properties of Tsallis entropies, J. Math. Phys. {\bf 47}, 023302 (2006).
  
  \bibitem[HA06a]{ha06a} F. Hansen, Extensions of Lieb's Concavity Theorem, J. Stat. Phys. {\bf 124}(2006), 87-101.
  
  \bibitem[HA06b]{ha06b} F. Hansen, Trace functions as Laplace transforms, J. Math. Phys. {\bf 47}, 043504 (2006).
   
  \bibitem[HZ14]{hazh} F. Hansen and Z. Zhang, Characterization of matrix entropies, arXiv:1402:2118v2, 20 Mar., 2014.
  
  \bibitem[HZ15]{hazh_v} F. Hansen and Z. Zhang, Characterization of matrix entropies, arXiv:1402:2118v3, 16 Mar., 2015.
  
  \bibitem[HP14]{HP} F. Hiai and D. Petz, {\em{Introduction to Matrix Analysis and Applications}}, Hindustan Book Agency and Springer Verlag, 2014.
  
  \bibitem[IS68]{IS} F. Itakura and S. Saito, Analysis synthesis telephony based on the maximum likelihood method,
  in 6th Int. Congr. Acoustics, Tokyo, Japan., pp. C-17-C-20 (1968)
  
  \bibitem[KL51]{KL} S. Kullback and R.A: Leibler, On information and sufficiency, Ann. Math. Statist. {\bf 22(1)}(1951), 79\,-\,86.
  
  \bibitem[LR99]{lr99} A. Lesniewski, M. B. Ruskai, Monotone Riemannian Metrics and Relative Entropy on Non-Commutative Probability Spaces, J. Math. Phys. {\bf 40}(1999), 5702-5724.
  
  \bibitem[LR74]{lr74} E. H. Lieb, M. B. Ruskai, Some Operator Inequalities of the Schwarz Type, Adv. in Math. {\bf 12}(1974), 269-273.
  
  \bibitem[LS14]{ls14} M. Lewin and J. Sabin, A Family of Monotone Quantum Relative Entropies, Lett. Math. Phys. {\bf 104}(2014) 691-705.
  
  \bibitem[LI74]{Linb74} G. Linblad, Expectations and Entropy inequalities, 
  Commun. Math. Phys. {\bf 39}(1974), 111-119.
  
  \bibitem[MA36]{Mah} P.C. Mahalonobis, On the generalized distance in statistics, Proceedings of National 
  Institute of Science of India, {\bf 12}(1936), 49\,-\,55. 
  
  \bibitem[NP05]{NP05} M. Nielsen and D. Petz, A simple proof of the strong subadditivity inequality,  Quantum Information \& Computation,  {\bf 6}(2005), 507\,-\,513.
 
 \bibitem[PE07]{pd07} D. Petz, Bregman divergence as relative operator entropy, Acta Math. Hungar, {\bf 116}(2007), 127-131.
 
 \bibitem[PV14]{PV} D. Petz and D. Virosztek, Some inequalities for quantum Tsallis entropy related to the strong subadditivity, Math. Inequal. Appl. {\bf 18(2)}(2015), 555-568.
 
 \bibitem[TR12]{tr12} J. A. Tropp, From joint convexity of quantum relative entropy to a concavity theorem of Lieb, Proc. Amer. Math. Soc. {\bf 140}(2012), 1757-1760.
 

 
 
\end{thebibliography}
\end{document}